\documentclass{fundam}

\publyear{22}
\papernumber{2102}
\volume{185}
\issue{1}

\usepackage{url} % takes care of hyperlinks, preferred over hyperref
\usepackage[ruled,lined]{algorithm2e}% provides Algorithm environment
\usepackage{graphicx}% allows for inclusion of graphic files (figures)
\usepackage{tikz}
\usetikzlibrary{automata, positioning, arrows}

\usepackage[T1]{fontenc}
\usepackage{graphicx}
\usepackage{amssymb}
\usepackage{amsmath}
\usepackage[utf8]{inputenc}
\usepackage{fontaxes}
\usepackage{microtype}
\usepackage{mathrsfs}
\usepackage{mathtools}
\usepackage{todonotes}
\usepackage{xspace}
\usepackage{xcolor}
\usepackage{complexity}
\usepackage{esvect}
\usepackage{xpatch}
\usepackage[all,defaultlines=3]{nowidow}
\usepackage[commandnameprefix=ifneeded]{changes}
\usepackage{mathrsfs}
\usepackage{subcaption}

\usepackage[colorlinks=true,linkcolor=blue,citecolor=blue]{hyperref}
\usepackage[nameinlink,capitalize]{cleveref}

\usepackage{color}

\newcommand{\Oof}{\mathcal{O}}

\newcommand{\Cc}{\mathscr{C}}

\newcommand{\N}{\mathbb{N}}

\renewcommand{\phi}{\varphi}

\renewcommand{\epsilon}{\varepsilon}

\newcommand{\tw}{\mathrm{tw}}

\newcommand{\TD}{\textsc{Treedepth}\xspace}
\newcommand{\td}{\mathrm{td}}

\newcommand{\WTD}{\textsc{Weighted Treedepth}\xspace}
\newcommand{\wtd}{\mathrm{wtd}}

\newcommand{\VC}{\textsc{Vertex Cover}\xspace}
\newcommand{\claimqed}{\hfill\ensuremath{\lhd}}

\usepackage[all]{xy}

\makeatletter
\xpatchcmd{\thmt@restatable}% Edit \thmt@restatable
{\csname #2\@xa\endcsname\ifx\@nx#1\@nx\else[{#1}]\fi}% Replace this code
{\ifthmt@thisistheone\csname #2\@xa\endcsname\ifx\@nx#1\@nx\else[{#1}]\fi\else\csname #2\@xa\endcsname\fi}% with this code
{}{} % execute code for success/failure instances
\makeatother

\newtheorem{observation}[definition]{Observation}

\newenvironment{claimproof}
  {\trivlist\PRstyle\item[]{\bfseries Proof of the claim:}\newline}{\claimqed\endtrivlist}

\tikzset{vertex/.style={circle, color=black, draw=black, align=center}}
\usetikzlibrary{calc}
\usetikzlibrary{shapes.geometric}

\begin{document}

\title{Weighted Treedepth is NP-complete on Graphs of Bounded Degree}

\address{University of Bremen, Bibliothekstraße 1, 28359 Bremen, Germany}

\author{Jona Dirks\\
University Clermont Auvergne, France\\
jona.dirks@uca.fr
\and Nicole Schirrmacher\\
University of Bremen, Germany\\
schirrmacher@uni-bremen.de
\and Sebastian Siebertz\\
University of Bremen, Germany\\
siebertz@uni-bremen.de
\and Alexandre Vigny\\
University Clermont Auvergne, France\\
alexandre.vigny@uca.fr} 

\maketitle

\runninghead{J. Dirks, N. Schirrmacher, S. Siebertz, A. Vigny}{Weighted Treedepth is NP-complete on Graphs of Bounded Degree}

\begin{abstract}
	A \emph{treedepth decomposition} of an undirected graph $G$ is a rooted forest~$F$ on the vertex set of~$G$ such that every edge $uv\in E(G)$ is in ancestor-descendant relationship in $F$. 
	Given a weight function $w\colon V(G)\rightarrow \N$, the \emph{weighted depth} of a treedepth decomposition is the maximum weight of any path from the root to a leaf, where the weight of a path is the sum of the weights of its vertices.
	It is known that deciding weighted treedepth is NP-complete even on trees.
	We prove that weighted treedepth is also NP-complete on bounded degree graphs.
	On the positive side, we prove that the problem is efficiently solvable on paths and on 1-subdivided stars.
\end{abstract}


\section{Introduction}

The structural parameters \emph{treewidth} and \emph{treedepth} lie at the heart of modern algorithmic and structural graph theory and play a pivotal role in both the theoretical understanding and the practical solution of hard computational~problems.

Treewidth, introduced independently in~\cite{bertele1972nonserial,halin1976s,robertson1984graph}, intuitively measures how close a graph is to being a tree. 
It underlies the celebrated graph minor theory, and many otherwise intractable problems are fixed-parameter tractable (FPT) when parameterized by treewidth, see~\cite{cygan2015parameterized}. In particular, Courcelle's theorem establishes that every graph property definable in monadic second-order logic~(MSO) can be decided in linear time on graphs of bounded treewidth, yielding a vast algorithmic meta-theory~\cite{courcelle1990monadic}.

In contrast, treedepth, introduced independently in~\cite{bodlaender1998rankings,katchalski1995ordered,nesetril2014sparsity,pothen1988complexity}, is a more restrictive measure, related to the height of forests embedding a graph's edges. 
It~captures not only tree-likeness but also nesting depth and hierarchical structure, making it particularly suited to problems involving recursive or hierarchical processes.
In graphs of bounded treewidth, many problems admit dynamic programming algorithms based on tree decompositions, whereas on graphs of bounded treedepth, recursive branching algorithms are often preferable and lead to more space-efficient solutions due to the shallow, hierarchical structure of treedepth decompositions, see e.g.~\cite{dvovrak2018parameterized,furer2017space,HegerfeldK20,nederlof2023hamiltonian,pilipczuk2021polynomial,pilipczuk2018space}. 
Furthermore, treedepth plays a crucial rule in the sparsity theory of Ne\v{s}et\v{r}il and Ossona de Mendez~\cite{nesetril2014sparsity}. 

Computing treewidth is NP-complete~\cite{arnborg1987complexity}, even on cubic graphs~\cite{BodlaenderBJKLM23} (hardness on graphs with maximum degree $9$ was established earlier in~\cite{BodlaenderT97}). 
We refer to the recent paper~\cite{bonnet2025treewidth} for a detailed overview of the exact, parameterized and approximation complexity of computing treewidth. 
It is one of the main open problems in the area whether treewidth can be computed in polynomial time on planar graphs. 
In contrast, the tightly related measure branchwidth can be computed efficiently on planar graphs~\cite{seymour1994call}.

Computing treedepth is also NP-complete~\cite{bodlaender1998rankings,pothen1988complexity}, even on chordal graphs~\cite{DereniowskiN06}. 
The problem is known to be polynomial-time solvable on permutation graphs, circular permutation graphs, interval graphs, circular-arc graphs, trapezoid graphs and on cocomparability graphs of
bounded dimension~\cite{deogun1994vertex}. 
The fastest known parameterized algorithms run in time $2^{\Oof(\td\cdot \tw)}\cdot n$, where $\td$ denotes the treedepth and $\tw$ denotes the treewidth of the input graph, when a tree decomposition is given with the input, and in time $2^{\Oof(\td^2)}\cdot n$, respectively, when no tree decomposition is given with the input~\cite{reidl2014faster}.
Recently, it was shown how to compute treedepth in time $2^{\Oof(\td^2)}\cdot n$ and $n^{\Oof(1)}$ space~\cite{NadaraPS22}, improving the algorithm of~\cite{reidl2014faster}, which uses exponential space.
It is an open question whether treedepth can be computed in polynomial time on graphs with bounded degree or on planar graphs, and it was recently conjectured that this is not the case~\cite{SchirrmacherSV25}.

In this work, we consider a weighted variant of treedepth, which, to the best of our knowledge was introduced in~\cite{DereniowskiN06}.
Weighted treedepth was independently rediscovered in~\cite{FominG0021} and recently found further applications in~\cite{jaffke2025parameterized}. 
A \emph{treedepth decomposition} of a graph $G$ is a rooted forest~$F$ on the vertex set of $G$ such that every edge $uv\in E(G)$ is in ancestor-descendant relationship in $F$.
We further assume that graphs are equipped with a weight function $w\colon V(G)\rightarrow \N$. 
The \emph{depth} of a treedepth decomposition is the number of vertices of a longest path from the root to a leaf.
The \emph{weighted depth} of a treedepth decomposition is the maximum weight of any path from the root to a leaf, where the weight of a path is the sum of the weights of its vertices.
It was shown in~\cite{DereniowskiN06} that deciding weighted treedepth is NP-complete even on trees. 

While it is open whether we can efficiently decide unweighted treedepth on planar graphs, the hardness result for weighted treedepth already on trees~\cite{DereniowskiN06} implies that, in particular, the weighted treedepth problem is hard on planar graphs.
The main technical result of our paper is to show that the problem is NP-hard on graphs with bounded degree, more precisely, on graphs with maximum degree~10. 
Note that containment of the problem in NP is trivial, since a treedepth decomposition is a polynomial-sized certificate. 
Our reduction is a non-trivial reduction from vertex cover and answers one of the open questions for treedepth in the weighted case negatively. 

On the positive side, we show that the problem is polynomial time solvable on paths, as well as on $1$-subdivided stars. 
The algorithms that we present for these severely restricted graph classes are, from a technical standpoint, not particularly difficult.
However, in light of the existing hardness results for the problem on trees, arguably one of the simplest nontrivial graph classes, and the current lack of understanding regarding the precise structural threshold at which the problem becomes computationally hard, our results represent the strongest algorithmic results that we have at this stage.

This leaves as an open question whether the problem is efficiently solvable on $2$-subdivided stars, and more generally, at what depth of trees the problem becomes NP-hard, and whether the problem is fixed-parameter tractable when parameterized by unweighted treedepth. 
It is an easy observation (see \cref{thm:fpt}) that it is fixed-parameter tractable when parameterized by the weighted treedepth itself.
A final question that remains open is whether weighted treedepth is NP-hard on graphs of degree smaller than 10, in particular, on cubic graphs.
% \rev{natural applications of weighted treedepth and importance. Seb. we gave three references [10,12,17]... I don't think we have to add more.}

\medskip\noindent\textbf{Organization.}
The paper is organized as follows.
After recalling the necessary notation in~\Cref{sec:prelim}, we provide algorithms to compute the \WTD in time FPT when parameterized by the \WTD itself and in polynomial time for paths and 1-subdivided stars in~\cref{sec:poly-cases}.
Our main contribution, the hardness result for \WTD on graphs of maximum degree~10, is presented in~\cref{sec:hardness}.

\section{Preliminaries}
\label{sec:prelim}

\noindent\textbf{Graphs.}
We consider finite and undirected graphs $G$ equipped with a weight function ${w\colon V(G)\to\mathbb{N}}$ where $\N$ is the set of positive integers (not including $0$).
A \emph{path} $P$ from a vertex~$v_0$ to a vertex $v_k$
is a sequence $(v_0,\ldots,v_k)$ of pairwise distinct vertices $v_0,\ldots,v_k\in V(G)$ such that $v_{i-1}v_{i}\in E(G)$ for $1\leq i\leq k$. The \emph{length} of the path $P$ is the number of edges $k$.
Two vertices~$u$ and~$v$ are \emph{connected} if there exists a path from $u$ to $v$.
A graph is \emph{connected} if every two of its vertices are connected.
A \emph{cycle} is a path in which the endpoints $v_0$ and $v_k$ are adjacent.
A \emph{vertex cover} is a vertex set $C\subseteq V(G)$ such that for every edge $uv\in E(G)$, we have $u\in C$ or $v\in C$.
The \emph{degree} of a vertex $v$ is the number of vertices adjacent to $v$.

\medskip\noindent\textbf{Trees and forests.}
A \emph{forest} $F$ is a graph that contains no cycles and a \emph{tree}~$T$ is a connected forest.
A \emph{rooted forest} $F$ is a forest in which each connected component (which is a tree) $T$ has a designated root vertex $r(T)$.
A \emph{leaf} of a forest is a vertex of degree 1.
The \emph{depth} of a tree is the length of a longest path from its root to a leaf and the depth of a forest is the maximum depth over its connected components.

A \emph{forest order} on a set $V$ is a partial order $\leq$ on $V$ such that for every $x \in V$, the set of elements smaller than $x$, $\{\, y \in V \mid y < x \,\}$, is linearly ordered by $\leq$. 
A \emph{tree order} is a forest order with a unique minimal element, called the \emph{root} of the order. 
Given a rooted forest~$F$, then for every connected component~$T$ of~$F$ with root~$r(T)$, we call a vertex $u\in V(T)$ an \emph{ancestor} of a vertex $v\in V(T)$ if $u$ lies on the unique path from $v$ to the root~$r(T)$, and write $u\leq_Fv$. The order $\leq_F$ is the \emph{forest order of $F$}. If $F$ is a tree, then $\leq_F$ is a tree order.
Note that vice versa, every forest order $\leq$ defines a unique forest $F$ such that the forest order of $F$ is equal to $\leq_F$. 
The {\em least common ancestor} of two vertices $u$ and $v$ is the $\le_F$-maximal vertex $x$ such that $x\le_F u$ and $x\le_F v$. If $u$ and $v$ are not in the same connected component of $F$, then the least common ancestor is not defined.
A \emph{star} $S_n$ is a tree consisting of a central root vertex adjacent to $n$ leaves.
A star where each edge is replaced by a path of length $n+1$ is called an \emph{$n$-subdivided star}.

\medskip\noindent\textbf{Treedepth decompositions.}
A \emph{treedepth decomposition} of a graph $G$ is a rooted forest~$F$ with $V(F)=V(G)$ such that for every edge $uv\in E(G)$, we have $u\leq_F v$ or $v\leq_F u$. 

\begin{observation}\label[observation]{obs:conn-decomposition}
  Let $G$ be a graph and let $H$ be a connected subgraph of $G$.
  Let~$T$ be a treedepth decomposition of $G$.
  Then there exists an element $x\in V(H)$ such that $x\leq_T y$ for all $y\in V(H)$, that is, $H$ has a unique $\leq_T$-minimal element~$x$.
\end{observation}

The \emph{weighted depth} of a treedepth decomposition is the maximum weight of any path from a root to a leaf (of the same tree) where the weight of a path is the sum of the weights of its vertices.

\begin{definition}
  The \emph{(weighted) treedepth} of a graph $G$, denoted $\wtd(G)$, is the minimum (weighted) depth over all treedepth decompositions of the graph $G$.
\end{definition}

\begin{observation}\label[observation]{obs:induced-subgraph}
  Consider a graph $G$, and a connected subgraph of $H\subseteq G$.
  Suppose $T$ is a treedepth decomposition of $G$, and let $h$ be the $\leq_T$-minimum element of $H$, which exists by \cref{obs:conn-decomposition}.
  Denote by $T_h$ the subtree of $T$ rooted at $h$.
  Then, the weighted depth of $T_h$ is at least the weighted treedepth of $H$.
\end{observation}

% \medskip\noindent\textbf{NP-hardness.}
% The class \textsf{NP} consists of problems $L\subseteq \Sigma^*$ (where $\Sigma$ is some fixed finite alphabet) that can be decided in polynomial time by a non-deterministic Turing machine, or equivalently, for which a given polynomial-size solution (called a certificate) can be verified in deterministic polynomial time.

% An instance of the weighted treedepth problem is (an encoding of) a graph~$G$ and a natural number $k$, and the question is whether the weighted treedepth of $G$ is at most~$k$.
% In our case, a polynomial-size certificate is a treedepth decomposition of weighted depth at most~$k$, which can obviously be verified to be a solution in polynomial time.

% A polynomial-time reduction from a problem $L\subseteq\Sigma^*$ to a problem $L'\subseteq\Sigma^*$ is a computable function $f\colon\Sigma^*\to\Sigma^*$ such that for all $w\in\Sigma^*$ we have $w\in L$ if and only if $f(w)\in L'$.
% A problem $L$ is \textsf{NP}\emph{-hard} if every problem in \textsf{NP} can be reduced to $L$ in polynomial time. 
% It is \textsf{NP}\emph{-complete} if it belongs to \textsf{NP} and is \textsf{NP}-hard.

\medskip\noindent\textbf{Fixed-parameter tractability.}
A parameterized problem $L\subseteq \Sigma^*\times \N$ is \emph{fixed-parameter tractable}~(FPT) if it can be solved in time $f(k)\cdot n^{\Oof(1)}$, where $k$ is the parameter of the instance (encoded in unary) and~$n$ the input size, for some computable function $f$. 

We assume that the integer weights assigned to the vertices of the graph are polynomially bounded in the number of vertices, that is, bounded by $n^c$ on $n$-vertex graphs for some fixed constant $c$, such that polynomial running times are both with respect to the number of vertices and with respect to the weights. 

\section{Efficient algorithms for weighted treedepth}
\label{sec:poly-cases}

\subsection{Fixed-parameter tractability with respect to solution size}

As a warm up, and introduction to the notion, we first present the following simple result. 
\begin{theorem}\label[theorem]{thm:fpt}
  Deciding \WTD is \textsf{NP}-hard even on chordal graphs, and is fixed-parameter tractable when parameterized by the weighted treedepth of the instance.
\end{theorem}

\begin{proof}
  The hardness holds by a reduction from classical \TD by simply setting the weight of every vertex to $1$, and using that \TD is \NP-hard on chordal graphs~\cite{DereniowskiN06}.
  
  For the FPT algorithm, let $G$ be a graph with weight function $w$. The $w$-blow-up of~$G$ is the graph~$G'$ obtained from $G$ by duplicating every vertex $x$ into a set~$V_x$ of $w(x)$ many vertices and turning each~$V_x$ into a clique. Additionally, for every edge $xy$ of $G$, we add the edges to turn $V_x,V_y$ into a bi-clique. In a treedepth decomposition, every vertex of a clique must appear in one root-to-leaf path, hence $\td(G')=\wtd(G)$. Since we can compute $\td(G')$ in time~$|G'|\cdot 2^{\Oof(\td(G')^2)}$, this computes $\wtd(G)$ in time $|G|^c\cdot 2^{\Oof(\wtd(G)^2)}$, where $c$ is a constant such that $w$ is bounded by $n^c$ (recall that we assume the weights are polynomially bounded in the number of vertices).%\qed
\end{proof}

Note that the reduction from \WTD to \TD does not preserve many structural parameters of $G$, such as planarity, exclusion of a minor, etc. 
So \NP-hardness of \WTD on some class~$\mathscr{C}$ might not yield hardness for \TD on this same class.  

\subsection{Weighted treedepth on paths}

We next show that the weighted treedepth of paths can be computed in polynomial time. 

\begin{theorem}\label[theorem]{thm:td-path}
	Given a weighted path $P$, an optimal treedepth decomposition can be computed in time~$O(|P|^3)$.
\end{theorem}

\begin{proof}
  We compute an optimum treedepth decomposition via a simple dynamic programming approach.
	The algorithm iteratively computes the \WTD of every subpath $P_I$ of $P=(v_0,\ldots, v_n)$. 
  Starting from the subpaths of length 0, all the way to the single subpath of length $n$. 
  Subpaths of length $0$ are trivially solved, as the only treedepth decomposition is to simply take that vertex, hence the weight of the optimal decomposition is the weight of this vertex.

	Given a subpath $P_{[i,j]}$ for some integers $i<j$,
	\[\wtd(P_{[i,j]})=\min\limits_{k\in\{i,\ldots,j\}}\left( w(v_k)+\max(\wtd(P_{[i,k-1]}),\wtd(P_{[k+1,j]}))\right).\]

	Given the order in which the optimal treedepth decomposition for every subpath is computed, the algorithm is given access to the \WTD of every subpath of the form $P_{[i,k]}$ with $i\le k<j$ and of the form $P_{[k,j]}$ with $i<k\le j$.
	The algorithm concludes by selecting the $v_k$ minimizing the resulting weight.
	
	For the running time, each subpath takes time at most $O(|P|)$, assuming that the arithmetic operations can be done in constant time. As there are at most $n^2$ many subpaths, this concludes the proof of \cref{thm:td-path}.%\qed
\end{proof}

\subsection{Weighted treedepth on 1-subdivided stars}
% \subsection{Weighted treedepth on subdivided stars}

Finally, we show that the weighted treedepth of 1-subdivided stars can be computed in polynomial time.

Note that the result does not follow directly from the existence of a small modulator to paths.
Recall that a modulator to a class $\Cc$ of graphs is a set of vertices whose removal results in a graph belonging to~$\Cc$.
In this manner, we obtain for example by \cref{thm:td-path} that weighted treedepth is polynomial-time solvable on the class of cycles.
In a cycle, we can guess the vertex that is the root of the decomposition (which in the running time leads to an additional factor of~$n$), and then are left with a path for which we can solve the problem efficiently.
However, this reasoning does not extend to the case of subdivided stars, even though they have a modulator of size~$1$.
The central vertex (i.e.~the modulator) must not necessarily be the root of the decomposition.

\begin{theorem}
  \WTD is polynomial-time solvable on 1-subdivided stars.
\end{theorem}

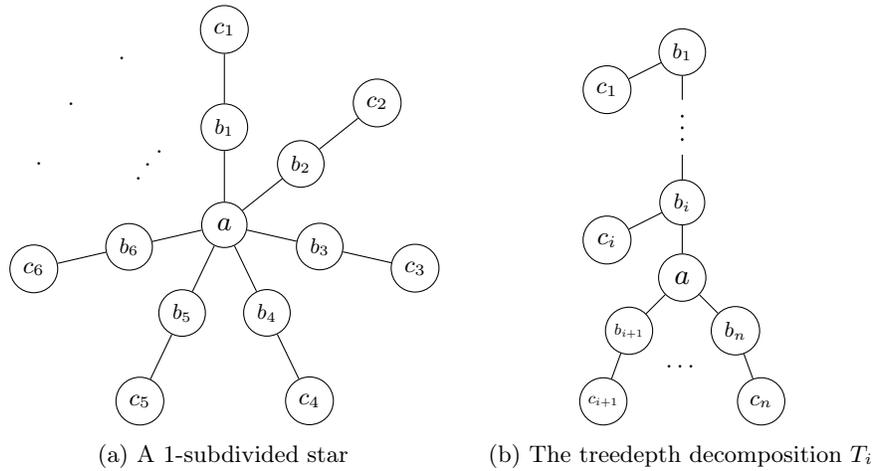
\begin{figure}[h]
	\begin{subfigure}{0.49\linewidth}
		\centering
		\begin{tikzpicture}[vertex/.style={circle,draw}]
			\node[vertex] (a) at (0,0) [scale=1.15] {$a$};
			\foreach \i [count=\j] in {0,...,5} {
				\node[vertex,scale=.9] (b\i) at (-\i*51.4+90:1.3) {$b_\j$};
				\draw (a) -- (b\i);
				\node[vertex] (c\i) at (-\i*51.4+90:2.6) {$c_\j$};
				\draw (b\i) -- (c\i);
			};
			
			\foreach \i in {-10,0,10} {
				\node[] (dot) at (-218.4+\i:1.3) {.};
			};
			\foreach \i in {-20,0,20} {
				\node[] (dot) at (-218.4+\i:2.6) {.};
			};
		\end{tikzpicture}
		\caption{A 1-subdivided star}
	\end{subfigure}
	\begin{subfigure}{0.49\linewidth}
		\centering
		\begin{tikzpicture}
			\node (a) at (0,-3) [circle,draw] [scale=1.2] {$a$};
			
			\node (b1) at (0,0) [circle,draw,scale=.9] {$b_1$};
			\node (db) at (0,-.75) {};
			\node (dd) at (0,-1.25) {};
			\node (bi) at (0,-2) [circle,draw,scale=.9] {$b_i$};
			
			\node (bi1) at ($(a) + (-135:1)$) [circle,draw,scale=.67] {$b_{i+1}$};
			\node (bn) at ($(a) + (-45:1)$) [circle,draw,scale=.9] {$b_n$};
			
			\node (c1) at (-1,-.5) [circle,draw] {$c_1$};
			\node (ci) at (-1,-2.5) [circle,draw,scale=1.05] {$c_{i}$};
			
			\node (ci1) at ($(bi1) + (-110:1)$) [circle,draw,scale=.7] {$c_{i+1}$};
			\node (cn) at ($(bn) + (-70:1)$) [circle,draw,scale=.95] {$c_n$};
			
			\node (d) at ($(0,0)+0.25*(bi1)+0.25*(bn)+0.25*(ci1)+0.25*(cn)$) {$\dots$};
			
			\path
			(b1) edge (db)
			(dd) edge (bi)
			(bi) edge (a)
			(bi1) edge (a)
			(bn) edge (a);

      \draw[line width = .35mm,dash pattern=on .05mm off 1.5mm,line cap=round] (0,-.85) to (0,-1.25);
			
			\path
			(c1) edge (b1)
			(ci) edge (bi)
			(ci1) edge (bi1)
			(cn) edge (bn);
		\end{tikzpicture}
		\caption{The treedepth decomposition $T_i$}
		\label{fig:weighted-star-tree-decomposition}
	\end{subfigure}
\caption{1-subdivided star and its treedepth decomposition}
\end{figure}

\vspace{-5mm}
\begin{proof}
  Let $G$ be a 1-subdivided star with weight function $w\colon V(G)\to\mathbb{N}$. We denote 
  \begin{itemize}
    \item the central vertex in $G$ by~$a$,
    \item the neighbors of $a$ by $b_1,\ldots,b_n$, and
    \item the private neighbor of $b_i$ by $c_i$ for every~$i\le n$.
  \end{itemize}
  Without loss of generality, the order satisfies $w(c_i)\geq w(c_j)$ for $i\leq j$, that is, the~$c_i$ are ordered by decreasing weights, and $w(b_i)\leq w(b_j)$ for $i\leq j$ if $w(c_i)=w(c_j)$, that is, the $b_i$ are ordered by increasing weights when the respective $c_i$ have equal weight.

  We now compute $n+1$ treedepth decompositions $T_0,\ldots, T_n$. We denote by~$d_i$ the weight of~$T_i$. 
  \begin{itemize}
    \item $T_0$ is the decomposition that is equal to $G$ with $a$ as the root. Hence, $a$ is minimal in~$\le_{T_0}$ and $a\le_{T_0} b_j\le_{T_0} c_j$ for $j\leq n$, while all other vertices are incomparable in the tree order~$\leq_{T_0}$. 
    \item $T_i$ for $i \geq 1$ is defined by $\leq_{T_i}$ as follows. We define $b_1\le_{T_i} \ldots \le_{T_i} b_i\le_{T_i} a$, as well as~$a\le_{T_i} b_j$ for $j>i$. Finally, $b_j\le_{T_i} c_j$ for all $j$. Unless to satisfy the transitivity property of $\le_{T_i}$, every other pair of vertices is incomparable. This decomposition is depicted in \cref{fig:weighted-star-tree-decomposition}.
  \end{itemize}
  
  We output $T_i$ with minimal depth $d_i$. 
  This algorithm has polynomial running time because it computes the weighted depth of $n+1$ treedepth decompositions.
  Next, we show that it computes indeed an optimal treedepth decomposition.
  
  Let $T$ be an optimal treedepth decomposition of $G$. 
  First, we show that all $c_i$ are leaves of~$T$.
  Assume that $c_i$ is not a leaf of $T$. 
  Then, we can simply remove~$c_i$ from the decomposition, attach its children at its parent (if $c_i$ was the root, then it has only one child which we make the new root), and add it as a leaf right below $b_i$.
  This is a treedepth decomposition as $c_i$ is connected only to $b_i$ and all other vertices keep the same order as in~$\leq_T$.
  Then, all root-to-leaf paths keep their value or lose the weight of $c_i$.
  The path from the root to the new leaf $c_i$ does not increase the depth of the decomposition, as the vertices of this new path were contained in the old path to $b_i$.
  Hence, we obtain a decomposition of at most the same weighted depth.

  Let $I\coloneqq\{i\leq n\ |\ b_i\leq_T a\}$ be the indices of the $b$-vertices above the vertex $a$ in $T$.
  We call $p_k$ the weight of the path in~$T$ from the root to the leaf $c_k$.

  To verify that the algorithm computes the optimal weighted treedepth, we prove the following two claims.
  First, we show that the vertices $b_i$ above the vertex $a$ in the treedepth decomposition are ordered (as induced by the weights of their $c_i$-neighbors).

  \begin{claim}
    For all $i<j\in I$, we have $b_i\le_T b_j$.
  \end{claim}

  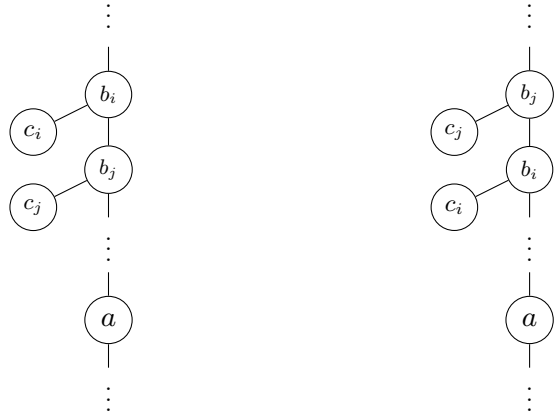
\begin{figure}[h]
    \begin{subfigure}[b]{.45\linewidth}
      \centering
      \begin{tikzpicture}
        \node (a) at (0,-4) [circle,draw,scale=1.2] {$a$};

        \node (bi) at (0,-1) [circle,draw,scale=.93] {$b_i$};
        \node (bj) at (0,-2) [circle,draw,scale=.89] {$b_j$};

        \node (ci) at (-1,-1.5) [circle,draw,scale=1.02] {$c_i$};
        \node (cj) at (-1,-2.5) [circle,draw,scale=.95] {$c_j$};

        \node (di) at ($(bi)+(0,.75)$) {};
        \node (dii) at ($(di)+(0,.4)$) {$\vdots$};

        \node (dj) at ($(bj)+(0,-.75)$) {};
        \node (djj) at ($(dj)+(0,-.2)$) {$\vdots$};
        \node (daj) at ($(a)+(0,.75)$) {};

        \node (da) at ($(a)+(0,-.75)$) {};
        \node (daa) at ($(da)+(0,-.2)$) {$\vdots$};

        \path
        (bj) edge (bi)
        (ci) edge (bi)
        (cj) edge (bj)
        (bi) edge (di)
        (a) edge (da)
        (a) edge (daj)
        (bj) edge (dj);
      \end{tikzpicture}
      \subcaption{Optimal treedepth decomposition $T$}
    \end{subfigure}
    \begin{subfigure}[b]{.45\linewidth}
      \centering
      \begin{tikzpicture}
        \node (a) at (0,-4) [circle,draw,scale=1.2] {$a$};

        \node (bi) at (0,-1) [circle,draw,scale=.89] {$b_j$};
        \node (bj) at (0,-2) [circle,draw,scale=.93] {$b_i$};

        \node (ci) at (-1,-1.5) [circle,draw,scale=.95] {$c_j$};
        \node (cj) at (-1,-2.5) [circle,draw,scale=1.02] {$c_i$};

        \node (di) at ($(bi)+(0,.75)$) {};
        \node (dii) at ($(di)+(0,.4)$) {$\vdots$};

        \node (dj) at ($(bj)+(0,-.75)$) {};
        \node (djj) at ($(dj)+(0,-.2)$) {$\vdots$};
        \node (daj) at ($(a)+(0,.75)$) {};

        \node (da) at ($(a)+(0,-.75)$) {};
        \node (daa) at ($(da)+(0,-.2)$) {$\vdots$};

        \path
        (bj) edge (bi)
        (ci) edge (bi)
        (cj) edge (bj)
        (bi) edge (di)
        (a) edge (da)
        (a) edge (daj)
        (bj) edge (dj);
      \end{tikzpicture}
      \subcaption{Treedepth decomposition $T'$}
    \end{subfigure}
    \caption{In the optimal treedepth decomposition $T$, the vertices above the central vertex $a$ must be ordered.}
    \label{fig:weighted-star-ordered}
  \end{figure}

  \begin{claimproof}
    Assume towards a contradiction that there is a treedepth decomposition $T'$ of weight smaller than the weight of $T$ where $b_j\leq_T b_i$ for $i<j\in I$.
    Without loss of generality, let $b_i$ be the child of $b_j$ in the treedepth decomposition $T'$.

    We show that the weighted depth of $T$ is smaller than the weighted depth of~$T'$ by inspecting the weight $p_k$ of every path from the root to the leaf $c_k$ in $T$ and finding a larger~$p'_{k'}$ in $T'$.

    \begin{enumerate}
      \item The paths from the root of $T$ to a leaf $c_k$ with $b_k <_T b_i$ have the same weights as these paths in~$T'$, i.e.~$p_k = p'_k$ for every $k$ with $b_k<_T b_i$.
      \item The path from the root of $T$ to the leaf $c_i$ has at most the weight of this path in $T'$:
      \begin{align*}
        p_i=\sum_{b_\ell\leq_T b_i}w(b_\ell)+w(c_i)
        \leq\sum_{b_\ell\leq_{T'} b_i}w(b_\ell)+w(c_i)
        =p'_i.
      \end{align*}
      \item The path from the root of $T$ to the leaf $c_j$ has at most the weight of the path from the root of $T'$ to the leaf $c_i$:
      \begin{align*}
        p_j=\sum_{b_\ell\leq_T b_j}w(b_\ell)+w(c_j)
        \leq\sum_{b_\ell\leq_{T'} b_i}w(b_\ell)+w(c_i)
        =p'_i.
      \end{align*}
      \item And lastly, the paths from the root of $T$ to a leaf $c_k$ with $b_k>_Tb_i$ have the same weights as these paths in $T'$, i.e.~$p_k=p'_k$ for every $k$ with $b_k>_T b_j$.
    \end{enumerate}

    Since every path from the root to a leaf in $T$ has at most the weight of a path from the root to a leaf in $T'$, the weighted depth of $T$ is at most the weighted depth of $T'$.%\claimqed
  \end{claimproof}

  Next, we show that if a vertex $b_j$ lies above $a$ in the treedepth decomposition, then every vertex $b_i$ with $i<j$ lies above $a$.

  \begin{claim}
    For all $i<j\leq n$, if $j\in I$, then $i\in I$.
  \end{claim}

  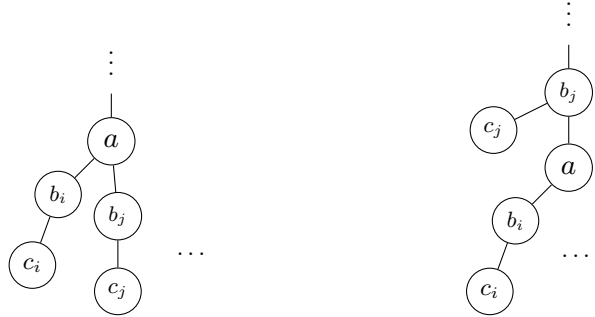
\begin{figure}[h]
    \begin{subfigure}[b]{.45\linewidth}
      \centering
      \begin{tikzpicture}
        \node (a) at (0,-3) [circle,draw,scale=1.2] {$a$};
        
        \node (bi) at ($(a) + (-135:1)$) [circle,draw,scale=.93] {$b_i$};
        \node (bj) at ($(a) + (-85:1)$) [circle,draw,scale=.89] {$b_j$};
        
        \node (ci) at ($(bi) + (-110:1)$) [circle,draw,scale=1.02] {$c_i$};
        \node (cj) at ($(bj) + (-90:1)$) [circle,draw,scale=.95] {$c_j$};

        \node (da) at ($(a)+(0,.75)$) {};
        \node (daa) at ($(da)+(0,.4)$) {$\vdots$};

        \node (d) at ($(1,0)+0.5*(bj)+0.5*(cj)$) {$\dots$};

        \path
        (bi) edge (a)
        (bj) edge (a)
        (ci) edge (bi)
        (cj) edge (bj)
        (a) edge (da);
      \end{tikzpicture}
      \subcaption{Optimal treedepth decomposition $T$}
    \end{subfigure}
    \begin{subfigure}[b]{.45\linewidth}
      \centering
      \begin{tikzpicture}
        \node (a) at (0,-3) [circle,draw,scale=1.2] {$a$};

        \node (bi) at ($(a) + (-135:1)$) [circle,draw,scale=.93] {$b_i$};
        \node (bj) at (0,-2) [circle,draw,scale=.89] {$b_j$};

        \node (ci) at ($(bi) + (-110:1)$) [circle,draw,scale=1.02] {$c_i$};
        \node (cj) at (-1,-2.5) [circle,draw,scale=.95] {$c_j$};
        
        \node (dj) at ($(bj)+(0,.75)$) {};
        \node (djj) at ($(dj)+(0,.4)$) {$\vdots$};

        \node (d) at ($(1,0)+0.5*(bi)+0.5*(ci)$) {$\dots$};

        \path
        (bi) edge (a)
        (bj) edge (a)
        (ci) edge (bi)
        (cj) edge (bj)
        (bj) edge (dj);
      \end{tikzpicture}
      \subcaption{Treedepth decomposition $T'$}
    \end{subfigure}
    \caption{In the optimal treedepth decomposition, the vertex $b_j$ only lies above the central vertex~$a$ if $b_i$ with $i<j$ lies above $a$.}
    \label{fig:weighted-star-above}
  \end{figure}

  \begin{claimproof}
    To show this claim, we actually prove that for all $i<j\leq n$, if $i\notin I$, then $j\notin I$.

    Assume towards a contraction that there exists a treedepth decomposition~$T'$ of smaller weighted depth with $I'\coloneqq\{\ell\leq n\ |\ b_\ell\leq_{T'} a\}$ where $i\notin I'$ for $i<j=\max\{I'\}$.
    
    We show that the weighted depth of $T$ is smaller than the weighted depth of~$T'$ by inspecting the weight $p_k$ of every path from the root to the leaf $c_k$ in $T$ and finding a larger $p'_{k'}$ in $T'$.
    
    \begin{enumerate}
      \item The paths from the root of $T$ to a leaf $c_k$ for $k\in I$ have the same weights as these paths in~$T'$, i.e.~$p_k=p'_k$ for every $k$ with $b_k<_T b_i$.
      \item The path from the root of $T$ to the leaf $c_j$ has at most the weight of the path from the root to the leaf $c_i$ in $T'$:
      \begin{align*}
        p_j&=\sum_{\ell\in I} w(b_\ell)+a+w(b_j)+w(c_j)\\
        &\leq\sum_{\ell\in I} w(b_\ell)+w(b_j)+a+w(b_i)+w(c_i)\\
        &=p'_i.
      \end{align*}
      \item The paths from the root of $T$ to a leaf $c_k$ for $k\notin I$, $k\neq j$ have at most the weight of these paths in $T'$:
      \begin{align*}
        p_k&=\sum_{\ell\in I} w(b_\ell)+a+w(b_k)+w(c_k)\\
        &\leq\sum_{\ell\in I} w(b_\ell)+w(b_j)+a+w(b_k)+w(c_k)\\
        &=p'_k.
      \end{align*}
    \end{enumerate}

    Since every path from the root to a leaf in $T$ has at most the weight of a path from the root to a leaf in $T'$, the weighted depth of $T$ is at most the weighted depth of $T'$.%\claimqed
  \end{claimproof}

  Therefore, the algorithm outputs the optimal weighted treedepth.
  This concludes this proof.%\qed
\end{proof}

\section{NP-hardness of weighted treedepth in graphs of bounded degree}
\label{sec:hardness}

We now come to the technically most demanding contribution of our paper and prove the following theorem. 

\begin{theorem}\label[theorem]{thm:td-degree}
  \WTD is \textsf{NP}-hard on graphs of maximum degree~10.
\end{theorem}

Recall that containment in \textsf{NP} is trivial, hence, the theorem implies \textsf{NP}-completeness of \WTD on graphs of maximum degree 10. 

To show the \textsf{NP}-hardness of weighted treedepth in graphs of bounded degree, we reduce from the vertex cover problem on graphs of maximum degree 3 to weighted treedepth in graphs of maximum degree 10. 
Recall that a \emph{vertex cover} of a graph $G$ is a vertex set that includes at least one endpoint of every edge of~$G$. \VC is \textsf{NP}-hard on cubic graphs~\cite{GareyJS76}.

% \begin{theorem}[Theorem 2.4 of~\cite{GareyJS76}]
%   \VC is \textsf{NP}-hard on cubic graphs.
% \end{theorem}

The rest of this section is devoted to the presentation of the reduction and its proof of correctness. 

Let $(G,k)$ be an instance of \VC, where all vertices of $G$ have degree $3$. We assume without loss of generality that the number of vertices $n:=|G|$ is a power of~$2$, adding isolated vertices if needed (at most doubling the size of the instance), and that $1<k<n-1$ because any set of $n-1$ vertices is a vertex cover.

\subsubsection*{Definition of the \WTD instance $(G',k')$.}

\mbox{}\\[3pt]
The graph~$G'$ is depicted in \cref*{picture-bdegree}. In detail, the instance $(G',k')$ is constructed as follows:
\begin{itemize}
	\item $G'$ contains two copies $v_1,v_2$ of every vertex $v\in V(G)$. We set $w(v_1)=w(v_2)=2$. 
	\item For every edge $uv\in E(G)$, we create the edges $u_1v_1$, $u_2v_2$, $u_1v_2$, and $u_2v_1$ in~$G'$.
	\item For every vertex $v\in V(G)$, we create three vertices $v_{g_1}$, $v_{g_2}$, and $v_{g_3}$. We set $w(v_{g_i}) = 1$ and connect each $v_{g_i}$ to both $v_1$ and $v_2$.
	\item We now create three binary trees $T_1$, $T_2$, and $T_3$ with exactly $n$ leaves (which is possible since~$n$ is a power of 2). For every $i\le 3$ and $v\in V(G)$, there is a leaf of $T_i$ named $v_{t_i}$.
	\item For every $v\in V(G)$, we connect $v_1$ and $v_{g_1}$ to $v_{t_1}$, $v_2$ and $v_{g_2}$ to $v_{t_2}$, and $v_{g_3}$ to~$v_{t_3}$.
	\item The weight of each root of the trees is $\beta = 3n+k+3$, then $w(y)=w(x)+3$ when $y$ is a child of $x$.
	\item Finally, let $\alpha$ be the sum of the weights of a leaf-to-root path in either $T_i$. \linebreak That is $\alpha=\sum\limits_{i=0}^{\log(n)}(\beta+3i)$. We set $k'=3n +k+\alpha+2$.
\end{itemize}

\newcommand{\createlabel}[2]{
${#1}$\\\textcolor{gray}{
	\hspace{0.5mm}#2}
}
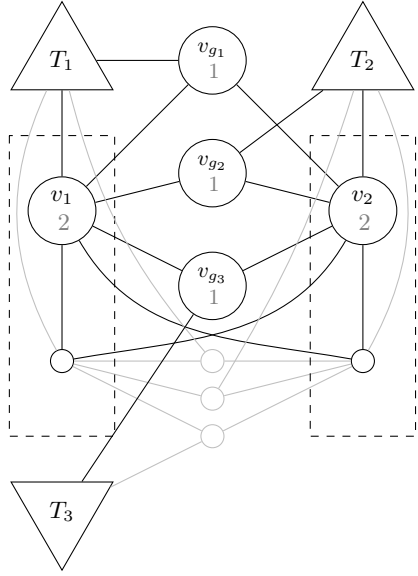
\begin{figure}
	\centering
	\begin{tikzpicture}
		\node[vertex] (v1) at (0,0) {\createlabel{v_1}{2}};
		
		\node[vertex] (v2) at (4,0) {\createlabel{v_2}{2}};
		
		\node[vertex,scale=.91] (vg1) at ($(v1)!0.5!(v2) + (0, 2)$) {\createlabel{v_{g_1}}{1}};
		\node[vertex,scale=.91] (vg2) at ($(v1)!0.5!(v2) + (0, 0.5)$) {\createlabel{v_{g_2}}{1}};
		\node[vertex,scale=.91] (vg3) at ($(v1)!0.5!(v2) + (0, -1)$) {\createlabel{v_{g_3}}{1}};
		
		\draw (v1) -- (vg1);
		\draw (v1) -- (vg2);
		\draw (v1) -- (vg3);
		
		\draw (v2) -- (vg1);
		\draw (v2) -- (vg2);
		\draw (v2) -- (vg3);

		\node[vertex] (u1) at (0,-2) {};
		
		\node[vertex] (u2) at (4,-2) {};
		
		\node[vertex, lightgray] (ug1) at ($(u1)!0.5!(u2) + (0, 0)$) {};
		\node[vertex, lightgray] (ug2) at ($(u1)!0.5!(u2) + (0, -0.5)$) {};
		\node[vertex, lightgray] (ug3) at ($(u1)!0.5!(u2) + (0, -1)$) {};
		
		\draw[lightgray] (u1) -- (ug1);
		\draw[lightgray] (u1) -- (ug2);
		\draw[lightgray] (u1) -- (ug3);
		
		\draw[lightgray] (u2) -- (ug1);
		\draw[lightgray] (u2) -- (ug2);
		\draw[lightgray] (u2) -- (ug3);
		
		\draw[] (v1) -- (u1);
		\draw[] (v2) -- (u2);
		
		\draw[] (v1) to[out=-60, in=170] (u2);
		\draw[] (v2) to[out=-120, in=10] (u1);
		
		\node[regular polygon, regular polygon sides = 3, draw=black] (t1) at (0,2) {$T_1$};
		
		\node[regular polygon, regular polygon sides = 3, draw=black] (t2) at (4,2) {$T_2$};
		
		\node[regular polygon, regular polygon sides = 3, rotate=180, draw=black] (t3) at (0, -4) {\textcolor{white}{$T_3$}};
		\node[] at (t3) (t3text) {$T_3$};
		
		\draw[draw=black, dashed] (-0.7,1) rectangle (0.7,-3);
		
		\draw[draw=black, dashed] (3.3,1) rectangle (4.7,-3);
		
		\draw (t1) -- (v1);
		\draw[lightgray] (t1) to[bend right=27] (u1);
		
		\draw (t2) -- (v2);
		\draw[lightgray] (t2) to[bend left=27] (u2);

		\draw (vg1) -- (t1);
		\draw[lightgray] (ug1) to[bend left=13] (t1);
		
		\draw (vg2) -- (t2);
		\draw[lightgray] (ug2) to[bend right=7] (t2);
		
		\draw (vg3) -- (t3);
		\draw[lightgray] (ug3) -- (t3);
	\end{tikzpicture}
	\caption{Result of the reduction when starting with two connected vertices. The names and weights are written exemplary.}\label{picture-bdegree}
\end{figure}

Note that every $v_{g_i}$ has exactly three neighbors and every vertex in $T_i$ has degree 3 (except for the root of degree 2). Furthermore, let $v$ be any vertex of~$G$ and $x,y,z$ its three neighbors.
Then, $v_1$ is only adjacent to $v_{t_1}$, the three $v_{g_i}$, $x_1,y_1,z_1$, as well as $x_2,y_2,z_2$. So~$v_1$ (and equivalently $v_2$) has degree 10, which is therefore the maximum degree of $G'$.

\subsubsection*{Correctness of the reduction, first direction.}
\begin{lemma}
	If there exists a vertex cover of size $k$ in $G$, then there exists a treedepth decomposition of weighted depth $k'$ for $G'$.
\end{lemma}
\begin{proof}
	Let $S$ be a vertex cover of size $k$ in $G$. 
	We construct the treedepth decomposition as follows: 
	The top of the treedepth decomposition is a path with the vertices $v_1$ and $v_2$ for $v \in S$ in any order. 
	Then for every~$v\notin S$, we continue the path by adding the vertices $v_{g_1}, v_{g_2}$ and $v_{g_3}$ in any order. 
	So far, this tree has a weighted depth of $4k+3(n-k)$ (since the two copies of each of the $k$ vertices of $S$ have weight~2 and the three middle vertices for the~$n-k$ vertices not in $S$ have weight~1).

	Observe that if we remove in $G'$ all vertices selected so far, the trees $T_1, T_2$ and~$T_3$ are in separate connected components. Additionally, if $T_1,T_2$ and $T_3$ are removed, all vertices are isolated. Thus, we continue the treedepth decomposition by attaching at the last selected vertex the trees $T_1$, $T_2$ and $T_3$ as subtrees. So far the decomposition has weighted depth $4k+3(n-k) + \alpha$. 
	
	For each vertex $v$, we either selected $v_1$ and $v_2$, leaving the $v_{g_i}$ only adjacent to $v_{t_i}$ respectively; or selected each $v_{g_i}$ and leaving $v_1$ and $v_2$ only adjacent to~$v_{t_1}$ and $v_{t_2}$ respectively. And since $S$ is a vertex cover, there are no connections between the remaining $v_i$. So each remaining vertex can be attached (in the treedepth decomposition) to the only leaf of the tree it is connected to. This increases the overall weighted depth by exactly 2 (since $k<n$), resulting in the demanded treedepth decomposition of depth $k'=4k+3(n-k)+\alpha+2$.%\qed
\end{proof}

\vspace{-3mm}
\subsubsection*{Correctness of the reduction, second direction.}
\begin{lemma}\label[lemma]{lem:small-td-small-vs}
	If there is a treedepth decomposition of weighted depth $k'$ in~$G'$, then there is a vertex cover of size $k$ in $G$.
\end{lemma}

The above lemma concludes the proof of \cref{thm:td-degree}. But before we prove it, we need to make some observations. 
Let $H$ be a subgraph of~$G'$ induced by all the vertices of $T_1$ and any subset of the vertices of $\{v_1 \mid v\in G\} \,\cup\, \{v_{g_1} \mid v\in G\}$. We call the vertices of $H$ not in $T_1$ the {\em dangling vertices}.
\begin{lemma}\label[lemma]{lem:tree-dec-of-trees}
	The weighted treedepth of $H$ is either $\alpha+w$ if there are no edges between two dangling vertices (where $w$ is the maximal weight of the dangling vertices), or at least~$\alpha+3$ if such an edge exists. Furthermore, every root-to-leaf path in an optimal treedepth decomposition has weight at least~$\alpha$.
\end{lemma}
\begin{proof}
	For a vertex $v$, we denote by $h(v)$ the distance of $v$ to a leaf in $T_1$ (which is invariant for all vertices of the same level of $T_1$ as $T_1$ is a complete binary tree). 
	The proof is done by induction on $h$ looking at every subtree $T_h$ rooted on a vertex $v$ with $h(v) = h$. We define $\alpha_h$ as $\sum\limits_{i=0}^{h}(\beta+3(\log(n)-i))$ and show that the weight depth of $T_h$ is at least $\alpha_h + 3$.
	
	If $h=0$, then $T_h$ is composed of one leaf of $T_1$ and up to two dangling vertices. If the dangling vertices are connected via an edge,\footnote{Note that there cannot be adjacent dangling vertices of weight 1.} then the three vertices must be in one path of weight $\alpha_0+1+2$. The case without edges is straightforward.

	By induction, let $T$ be an optimal treedepth decomposition for $T_h$. Assume the existence of an edge between two dangling vertices (and therefore the existence of a dangling vertex of weight~$2$).
	Assume first that the root of the optimal treedepth decomposition $T$ is the root of $T_h$. Either the existing edges keep the graph connected (and further deletions must occur before reducing to~$T_{h-1}$), or there is one subtree of~$T_h$ that, by induction, has weighted treedepth at least $\alpha_{h-1}+3$. In both case, adding the weight of the root (i.e.~$\beta+3(\log(n)-h)$) makes the weight of $T$ at least $\alpha_{h-1}+\beta +3(\log(n)-h) +3 =\alpha_h+3$.

	Assume now that the root is a vertex $a$ of $T_h$ (that is not a dangling vertex).
	Note that the weight of $a$ is at least $\beta+3(\log(n) - h+1)$.
	Then by induction, there is a subtree $T_{h-1}$ of $T_h$ not containing $a$, which has weighted treedepth at least $\alpha_{h-1}$. Adding the weight of $a$ makes at least $\alpha_{h-1} + \beta +3(\log(n) - h+1) +3$, hence, at least $\alpha_h+3$.
	
	Assume finally that the root of $T$ is a dangling vertex $a$. We can assume that $a$ is not of degree $1$, otherwise switching it with its only neighbor does not increase the weight of $T$, and we go to one of the previous cases. So $a$ is adjacent to a dangling vertex $b$ with $w(a)+w(b)\ge 3$. Even if removing $a$ removes all edges between dangling vertices, the weighted depth of $T-a$ must be at least $\alpha_h+w(b)$, so the weighted treedepth of $T$ is at least $\alpha_h+w(a)+w(b)$, hence at least $\alpha_h+3$.%\qed
\end{proof}

Of course, \cref{lem:tree-dec-of-trees} also applies to $T_2$. The case of $T_3$ is simpler, the dangling vertices are independent and each of weight $1$. For $T_3$, we only use that its weighted treedepth is at least $\alpha$.
We now go to the proof of \cref{lem:small-td-small-vs}.

\begin{proof}%[of \cref{lem:small-td-small-vs}]
	Let $T$ be a treedepth decomposition of $G'$ with weighted depth~$k'$.
	Let~$\le_T$ be the tree order of $T$. 
	By \cref{obs:conn-decomposition}, the trees $T_1$, $T_2$, and~$T_3$ have unique \mbox{$\leq_T$-minimal} elements $t_1, t_2, t_3$.
	\begin{claim}
		$t_1$, $t_2$, and $t_3$ are pairwise $\le_T$-incomparable.
	\end{claim}
	\begin{claimproof}
		Assume towards a contradiction and without loss of generality that $t_1<_T t_2$, which also implies that $t_1<_T y$ for every $y\in V(T_2)$. 
		By \cref{lem:tree-dec-of-trees}, $T_2$ has weighted treedepth at least $\alpha$. 
		By \cref{obs:induced-subgraph}, the subtree of $T$ rooted at $t_2$ has weighted depth at least $\alpha$.
		Every element of $T_1$, and in particular $t_1$, has weight at least $\beta$. 
		As $t_1<_T t_2$, $T$ has weighted depth at least $\alpha+\beta$. 
		However, note that $\alpha+\beta = \alpha +3n+k+3 > \alpha + 3n +k+2 = k'$, which contradicts that the weighted \linebreak depth of $T$ is~$k'$.%\claimqed
	\end{claimproof}
	
	The least common ancestors of $(t_1,t_2)$ and $(t_2,t_3)$ must be $\le_T$-comparable (they are both ancestors of $t_2$). Let $a$ be the $\le_T$-maximal of the two. Again $a$ must be $\le_T$-comparable with the least common ancestor of $(t_1,t_3)$, since either both are ancestors of $t_1$ or both of $t_3$. Let $b$ be the $\le_T$-maximal of the two. Note that no vertex of any $T_i$ is $\le_T$-smaller than $b$.
	
	Let $S'$ be the set of vertices on the path from $b$ to the root of $T$. Observe that in $G' - S'$, the trees~$T_1$,~$T_2$, and~$T_3$ are pairwise in different connected components. This is due to the fact that any path from a vertex of $T_i$ to a vertex of $T_j$ must go through a vertex $\le_T$-smaller than the least common ancestor of $t_i$ and $t_j$, and hence $\le_T$-smaller than $b$.
	
	Now consider the following set $S = \{v \in V(G) \mid v_1 \in S' \textrm{ or } v_2 \in S'\}$. We claim that $S$ is a vertex cover of size at most $k$.
	\begin{claim}
		$S$ is a vertex cover of size at most $k$.
	\end{claim}
	\begin{claimproof}
		First, we show that $|S| \leq k$.
		Assume towards a contradiction that $S'$ contains every vertex of $G'$ with weight $2$. Then the weighted depth of $T$ is at least $4n+\alpha+1$ (where the $+1$ is due to the remaining dangling vertices) and since $k<n-1$ we have that $4n+\alpha+1 >3n+k+\alpha +2$. So $S'$ does not contain every vertex of weight $2$ and so in $G'-S'$ there are dangling vertices of weight $2$. Hence, with \cref{lem:tree-dec-of-trees} and since the weighted depth of $T$ is at most $3n+k+\alpha+2$, we obtain that $w(S')\le 3n+k$.

		Now if $v\in S$ then either:
		\begin{itemize}
			\item $v_1,v_2 \in S'$, or
			\item $v_1\not\in S'$ and then $v_2$ in $S'$ by definition, but also $v_{g_2},v_{g_3} \in S'$ (since we must separate $T_1$ from $T_2$ and $T_3$), or
			\item $v_2\not\in S'$ and then $v_1\in S'$ but also $v_{g_1},v_{g_3} \in S'$ (since we must separate $T_2$ from $T_1$ and $T_3$).
		\end{itemize} 		
		In all of these cases, a weight of $4$ can be charged to $v$. Further, if $v\notin S$ then $v_{g_1},v_{g_2},v_{g_3}\in S'$. Thus, $w(S') = 4|S| + 3(n-|S|)=3n+|S|$. It follows that~$|S| \leq k$.
		
		For showing that $S$ is a vertex cover, assume towards a contraction that there exists an edge $uv$ with both $u$ and $v$ not in $S$. Then in particular, none of $u_1,u_2,v_1,v_2$ is in $S'$. However, then $S'$ would not separate $T_1$ and $T_2$.%\claimqed
	\end{claimproof}
	This concludes the proof of \cref{lem:small-td-small-vs}.%\qed
\end{proof}

\section{Conclusion}

We studied weighted treedepth, the weighted version of the well-known treedepth parameter.
Our results show that \WTD can be solved in polynomial time on certain simple graphs such as paths and 1-subdivided stars.
While these results may initially seem modest, it is important to note that computing \WTD is \textsf{NP}-hard on trees~\cite{DereniowskiN06}.
In this paper, we provide further insight by demonstrating that \WTD is also \textsf{NP}-hard on graphs with a maximum degree of~10, thus completing a significant part of the complexity landscape.

The main open problem is now the question if weighted treedepth is fixed-parameter tractable parameterized by (unweighted) treedepth.
Additionally, the proof of the conjectures that treedepth is \textsf{NP}-hard on planar graphs and graphs of bounded degree remains an open challenge.
Another interesting question is to what factors weighted treedepth can be approximated.

\bibliographystyle{fundam}
\bibliography{ref.bib}

%%%%%%%%%%%%%%%%%%%%%%%%%%%%%%%%%%%%%%%%%%%%%%%%%%%%%%%%%%%%%%%%%%%%%%

\end{document}